\documentclass[12pt]{article}

\usepackage{amsmath, amssymb, amsthm, amsfonts,thmtools}

\usepackage[textsize=tiny]{todonotes}

\newtheorem{theorem}{Theorem}[section]
\newtheorem{lemma}[theorem]{Lemma}
\newtheorem{meta-theorem}[theorem]{Meta-Theorem}
\newtheorem{claim}[theorem]{Claim}
\newtheorem{conjecture}[theorem]{Conjecture}

\newtheorem{corollary}[theorem]{Corollary}

\newtheorem{observation}[theorem]{Observation}
\newtheorem{definition}[theorem]{Definition}
\newtheorem*{llll}{Lov\'asz Local Lemma (LLL)}

\definecolor{darkgreen}{rgb}{0,0.5,0}
\usepackage{hyperref}
\hypersetup{
    unicode=false,          
    colorlinks=true,        
    linkcolor=red,          
    citecolor=darkgreen,        
    filecolor=magenta,      
    urlcolor=cyan           
}
\usepackage[capitalize, nameinlink]{cleveref}
\crefname{theorem}{Theorem}{Theorems}
\Crefname{lemma}{Lemma}{Lemmas}
\Crefname{conjecture}{Conjecture}{Conjectures}
\Crefname{claim}{Claim}{Claims}
\Crefname{remark}{Remark}{Remarks}
\Crefname{observation}{Observation}{Observations}

\newcommand{\eps}{\varepsilon}

\newcommand{\poly}{\operatorname{\text{{\rm poly}}}}

\newcommand{\prob}[1]{P(#1)}

\newcommand{\evt}{\mathcal E}
\newcommand{\vbl}{\operatorname{vbl}}
\newcommand{\alg}{\mathcal A}

\newcommand{\srep}{S_\textrm{rep}}
\newcommand{\snon}{S_\textrm{non}}
\newcommand{\isnon}{{S^{\textrm{o}}_\textrm{non}}}

\newcommand{\R}{\mathbb{R}}

\begin{document}
\date{}
\title{Generalizing the Sharp Threshold Phenomenon for the Distributed Complexity of the Lovász Local Lemma}
\author{Sebastian Brandt\\
\small ETH Zurich \\
\small brandts@ethz.ch\\
\and
Christoph Grunau\\
\small ETH Zurich \\
\small grunau@student.ethz.ch\\
\and 
V\'aclav Rozho\v n\\
\small ETH Zurich \\
\small rozhonv@ethz.ch\\
}
\maketitle

\begin{abstract}
	Recently, Brandt, Maus and Uitto [PODC'19] showed that, in a restricted setting, the dependency of the complexity of the distributed Lov\'asz Local Lemma (LLL) on the chosen LLL criterion exhibits a sharp threshold phenomenon: They proved that, under the LLL criterion $p2^d < 1$, if each random variable affects at most $3$ events, the deterministic complexity of the LLL in the LOCAL model is $O(d^2 + \log^* n)$. In stark contrast, under the criterion $p2^d \leq 1$, there is a randomized lower bound of $\Omega(\log \log n)$ by Brandt et al. [STOC'16] and a deterministic lower bound of $\Omega(\log n)$ by Chang, Kopelowitz and Pettie [FOCS'16].
	Brandt, Maus and Uitto conjectured that the same behavior holds for the unrestricted setting where each random variable affects arbitrarily many events.
	
    We prove their conjecture, by providing an algorithm that solves the LLL in time $O(d^2 + \log^* n)$ under the LLL criterion $p2^d < 1$, which is tight in bounded-degree graphs due to an $\Omega(\log^* n)$ lower bound by Chung, Pettie and Su [PODC'14].
    By the work of Brandt, Maus and Uitto, obtaining such an algorithm can be reduced to proving that all members in a certain family of functions in arbitrarily high dimensions are convex on some specific domain.
    Unfortunately, an analytical description of these functions is known only for dimension at most $3$, which led to the aforementioned restriction of their result.
    While obtaining those descriptions for functions of (substantially) higher dimension seems out of the reach of current techniques, we show that their convexity can be inferred by combinatorial means. 
\end{abstract}

\maketitle

\section{Introduction}
\subsection{Background}
The Lov\'asz Local Lemma is a celebrated result from 1975 due to Erd{\H o}s and Lov\'asz \cite{erdos75local}, with applications in many types of problems such as coloring, scheduling or satisfiability problems \cite{alonspencer,Beck1991,SuLLL2017,Czumaj2000,czumaj00algorithmic,elkin15matching,haeupler10constructive,leighton99fast,Moser09}. It states the following.
\begin{llll}
    Let $\{X_1, \dots, X_m\}$ be a set of mutually independent random variables and $\evt_1, \dots, \evt_n$ probabilistic events that depend on the $X_i$. For each $\evt_i$, let $\vbl(\evt_i)$ denote the random variables $\evt_i$ depends on. We say that $\evt_i$ and $\evt_j$ share a random variable if $\vbl(\evt_i) \cap \vbl(\evt_j) \neq \emptyset$. Assume that there is some $p < 1$ such that for each $1 \leq i \leq n$, we have $\prob{\evt_i} \leq p$ , and let $d$ be a positive integer such that each $\evt_i$ shares a random variable with at most $d$ other $\evt_j$ (where $j \neq i$). Then, if $4pd \leq 1$, there exists an assignment of values to the random variables such that none of the events $\evt_i$ occurs.\footnote{We note that the \emph{LLL criterion} $4pd \leq 1$ guaranteeing the existence of the desired variable assignment is not optimal and has been subject to improvements by Spencer \cite{Spencer77} and Shearer \cite{Shearer85}.}
\end{llll}

The LLL can be seen as a generalization of the well-known fact that for any set of \emph{independent} events that all occur with probability strictly less than $1$, the probability that none of the events occurs is non-zero: some amount of dependency between the events is tolerable for preserving the avoidance guarantee---how much exactly depends on the parameter $p$ that bounds the occurrence probabilities of the events.

While being an indispensable tool for applying the probabilistic method, the LLL, in its original form, is of limited usefulness if seen from an algorithmic standpoint, as it gives a purely existential statement and does not provide a method for finding such an assignment to the random variables.
The underlying algorithmic question of computing such an assignment, called the \emph{algorithmic (or constructive) LLL (problem)} received considerable attention in a series of papers \cite{Alon1991,Czumaj2000,Molloy1998,Moser08,Moser09,Srinivasan2008}, starting with Beck \cite{Beck1991} in the 90s, and culminating in a breakthrough result by Moser and Tardos \cite{MoserTardos10} in 2010.
The latter work showed that an assignment to the random variables that avoids all events can be found quickly by a simple resampling approach.
Moreover, this approach is easily parallelizable, and implies a (randomized) distributed algorithm that finds the desired assignment in $O(\log^2 n)$ rounds of communication in a distributed setting w.h.p.\footnote{As usual, w.h.p. stands for ``with probability at least $1-1/n$".} for the LLL criterion $ep(d+1) < 1$.
In the following, we take a closer look at the distributed version of the algorithmic LLL, the main topic of this work.

\paragraph{The Distributed LLL}
Let an instance of the LLL be given by mutually independent random variables $X_1, \dots, X_m$ and events $\evt_1, \dots, \evt_n$ satisfying some LLL criterion that guarantees existence of an assignment avoiding all events.
The distributed version of the LLL is commonly phrased using the notion of the so-called \emph{dependency graph}.
In the dependency graph of an LLL instance, the events $\evt_i$ are the nodes, and there is an edge between two events $\evt_i, \evt_j$ if the two events share a variable.
Each node $\evt_i$ is aware of $\vbl(\evt_i)$ and knows for exactly which combinations of values for the random variables in $\vbl(\evt_i)$ the event $\evt_i$ occurs.
As before, the task is to find an assignment to the variables such that none of the events occurs.
To specify the output, each node $\evt_i$ has to output a value for each variable it depends on, and any two nodes outputting a value for the same random variable have to agree on the value.

We will consider the LLL in the LOCAL model\footnote{The communication graph for the LLL is the dependency graph. For details regarding the LOCAL model, we refer to \cref{sec:model}.} of distributed computing \cite{linial92,peleg00}, in which the LLL has been the focus of a number of important works in recent years (see \cref{sec:related} for an in-depth overview).
One particularly intriguing result underlining the importance of the LLL was given by Chang and Pettie \cite{ChangHierarchy17}: they show that \emph{any} problem from a very natural problem class, called \emph{locally checkable labelings}\footnote{Roughly speaking, these are problems for which the correctness of the global solution can be verified by checking the correctness of the output in the local neighborhood of each node.}, that has sublogarithmic randomized complexity also admits a randomized algorithm that solves it in time $T_{\operatorname{LLL}}(n)$, where $T_{\operatorname{LLL}}(n)$ denotes the randomized complexity of the LLL under a polynomial criterion (i.e., a criterion of the form $pd^c \in O(1)$ for an arbitrarily large constant $c$).

\paragraph{The LLL criterion}
As can be expected, the complexity of solving the algorithmic LLL depends on the chosen LLL criterion.
Strengthening the LLL criterion, i.e., restricting\footnote{There is some disagreement about whether this should be called a strengthening or weakening of the LLL criterion. We will use the same (sensible) terminology as the closest work to ours by Brandt, Maus and Uitto \cite{brandt2019sharp}, to make it easier to relate the results.} the set of allowed LLL instances by making fewer pairs $(p, d)$ satisfy the criterion, clearly can only reduce the complexity of the LLL problem, but it is a major open question precisely how the LLL complexity relates to the chosen criterion.
At which points in strengthening the LLL criterion does the asymptotic complexity of the LLL change and do we obtain smooth or sharp transitions between the different complexities?

While it is known, due to a result by Chung, Pettie and Su \cite{SuLLL2017}, that $\Omega(\log^* n)$ rounds are required for any LLL criterion, the only lower bounds known so far that could possibly be used to obtain a separation between the complexities for different criteria are an $\Omega(\log \log n)$ lower bound for randomized algorithms by Brandt et al.\ \cite{LLL_lowerbound}, and an $\Omega(\log n)$ lower bound for deterministic algorithms by Chang, Pettie and Kopelowitz \cite{Chang2016a}, which both hold even under the strong criterion $p2^d \leq 1$.
It is natural to ask whether any further strengthening of the LLL criterion breaks the lower bound or whether the lower bound can be extended to stronger criteria.

Very recently, Brandt, Maus and Uitto \cite{brandt2019sharp} showed that if we restrict the random variables to affect at most $3$ events each (which they call \emph{rank} at most $3$), then already under the minimally strengthened criterion $p2^d < 1$, there is a deterministic LLL algorithm with a complexity of $O(\poly d + \log^* n)$.
They conjectured that this behavior also holds without their restriction on the variables.

\begin{conjecture}[\cite{brandt2019sharp}, rephrased]\label{conj:bmu}
    There is a (deterministic) distributed algorithm that solves the LLL problem in time $O(d^2+ log^*n)$ under the criterion $p2^d < 1$.
\end{conjecture}

\subsection{Contributions and Techniques}

In this work, we prove \cref{conj:bmu}, by providing such a deterministic algorithm.
This gives a first (unrestricted) answer to the aforementioned question about the relation between the LLL criterion and the complexity of the LLL:
a sharp transition occurs at the criterion $p2^d < 1$, where the complexity of the LLL drops from $\Omega(\log \log n)$ randomized, resp.\ $\Omega(\log n)$ deterministic, to $O(d^2 + \log^* n)$.
Moreover, our upper bound is tight on bounded-degree graphs due to the $\Omega(\log^* n)$ lower bound by Chung, Pettie and Su \cite{SuLLL2017}.
Finally, as is the nature of upper bounds for the LLL, our result immediately implies the same upper bound for all problems that can be phrased as an LLL problem with criterion $p2^d < 1$, such as certain hypergraph edge-coloring problems or orientation problems in hypergraphs (see \cite{brandt2019sharp}).

\label{subsec:techniques}
\paragraph{Previous Techniques}
Our work builds on techniques developed in \cite{brandt2019sharp}.
In their work, Brandt, Maus and Uitto obtain an $O(d^2 + \log^* n)$-round LLL algorithm under the criterion $p2^d < 1$ for the case of variables of rank at most $3$. In the following, we give an informal overview of their approach. 

The basic idea of the algorithm is to go sequentially through all variables and fix them to some values one by one while preserving a certain invariant that makes sure that the final assignment avoids all events.
In order to define the invariant, each edge of the dependency graph is assigned two non-negative values, one for each endpoint of the edge, that sum up to at most $2$.  
When fixing a random variable, the algorithm is also allowed to change these ``book-keeping" values.
The invariant now states that for any node $v$ in the dependency graph, the product of the $\deg(v)$ values around $v$ multiplied by $p$ is an upper bound for the conditional probability of the event $\evt_v$ associated with node $v$ to occur (where we naturally condition on the already-fixed random variables being fixed as prescribed by the (partial) value assignments performed by the algorithm so far).
If this invariant is preserved, then, after all variables are fixed, each event $\evt_v$ occurs with probability at most $2^{\deg(v)} \cdot p \leq p2^d < 1$, and therefore with probability $0$, as desired.

Brandt, Maus and Uitto do not only show that such a sequential process preserving the invariant at all times exists (even if the order in which the random variables have to be fixed are chosen adversarially), but also that it can be made to work in a local manner: in order to fix a random variable, the algorithm only needs to know the random variables and edge values in a small local neighborhood.
This allows to process random variables that affect events that are sufficiently far from each other in the dependency graph in parallel.
By adding an $O(\log^* n)$-round preprocessing step to the algorithm where a $2$-hop node coloring with $O(d^2)$ colors is computed in the dependency graph, the sequential fixing process can then be parallelized by iterating through the color classes in a standard way, yielding the desired runtime of $O(d^2 + \log^* n)$ rounds.
We will provide a more detailed overview of the algorithm from \cite{brandt2019sharp} in \cref{sec:prevalgo}.

The crucial, and rather surprising, observation making the algorithm work is that in each step in which a random variable is fixed, the existence of a value for that random variable that preserves the invariant is guaranteed if a certain function is shown to be convex on some domain.
Hence, proving the existence of the desired algorithm is reduced to solving an analytical problem for a fixed function $f$, providing a very intriguing connection between distributed algorithms and analysis.
To be precise, Brandt, Maus und Uitto show that for any integer $r \geq 2$, there is a fixed function $f_r : D \rightarrow \R$ on some domain $D \subset \R^{r-1}$ satisfying the following property: if $f_r$ is convex, then for any rank-$r$ random variable, there is a value that this variable can be fixed to such that the invariant is preserved.
By proving the convexity of $f_3(a,b) = 4 + 1/2 \cdot (ab - 2a - 2b - \sqrt{ab(4 - a)(4 - b)})$, they prove the desired upper bound for the case of variables of rank at most $3$.\footnote{Taking care of the case of rank-$1$ and rank-$2$ variables is comparably easy.}

One of the main problems with extending this proof to arbitrary ranks is that the function is only given in an indirect way, by a characterization of the set of points in $\R^r$ that lie below and on the function.
No closed-form expression describing $f_r$ is known for any $r > 3$, and the relatively compact form of the function for the case $r = 3$ is arguably due to the cancellation of certain terms that do not cancel out in higher dimensions.
In fact, none of the ways to obtain $f_3$ from the characterization of the mentioned point set seems to yield \emph{any} closed-form expression if adapted to higher dimensions, and even if a closed-form expression for all $f_r$ were found in some way, it is far from clear that proving convexity of these functions would be feasible.

\paragraph{New Techniques}
We overcome this obstacle by showing that, perhaps surprisingly, even without any analytical access to the functions $f_r$, we can infer their convexity for all $r$.
In the following we give an informal overview of our approach.
Our main idea is to prove convexity of $f_r$---or equivalently, convexity of the set bounded by $f_r$ from below--- by finding a so-called locally supporting hyperplane for each point $q$ on $f_r$.
More precisely, for each such $q$, we want to find a number of vectors such that the following two properties hold:
\begin{enumerate}
    \item\label{item:hyperplane} The affine subspace of $\R^r$ spanned by the vectors and containing $q$ is a hyperplane, i.e., an affine subspace of dimension $r-1$.
    \item\label{item:ball} In an $\varepsilon$-ball around $q$, the hyperplane is contained in the set consisting of all points on and below $f_r$.
\end{enumerate}
These properties ensure convexity of $f_r$ in $q$; however, a priori it is completely unclear how to find such vectors.
In order to obtain these vectors, we consider the combinatorial description of the points on and below $f_r$ that is tightly connected to the aforementioned invariant:
Consider a hyperedge of rank $r$ and write two non-negative values that sum up to at most $2$ on each edge of the skeleton of the hyperedge (i.e., a clique induced by the hyperedge) one value for each endpoint of the edge.
For each endpoint of the hyperedge multiply the $r-1$ values belonging to the endpoint, and consider the $r$-dimensional vector obtained by collecting the resulting products.
The points that can be generated in this way are exactly the (non-negative) points that lie on or below $f_r$.

For each such point $q'$, call the tuple of the $\Theta(r^2)$ values written on the edges that generate $q'$ in the above description a generator of $q'$; a point can have (and usually has) more than one generator.
Roughly speaking, we find the desired vectors for a point $q$ by picking an arbitrary generator and, for each edge $e$ in the skeleton of the hyperedge, computing the vector by which $q$ changes if we subtract some small $\varepsilon$ from one value on $e$ and add it to the other.
A crucial insight is that it is fine to pick such a large set of $\Theta(r^2) \gg r - 1$ vectors: due to the specific construction, one can show that the affine subspace spanned by these $\Theta(r^2)$ vectors and containing $q$ is $(r-1)$-dimensional.
Moreover, the redundancy contained in this choice enables us to prove \cref{item:ball} by finding, for each $q'$ on the hyperplane in an $\varepsilon$-ball around $q$, a way to write $q' - q$ as a linear combination of $r - 1$ of these vectors that satisfies certain desirable properties.

Note that we will use terminology that does not refer to the convexity of the function $f_r$ as we do not make use of this function from an analytical perspective. 
Instead, we will aim for the equivalent goal of showing that the set bounded by $f_r$ from below is convex, by making use of its combinatorial description. 

\subsection{Further Related Work}\label{sec:related}
Following the resampling approach of Moser and Tardos \cite{MoserTardos10}, many of the results for the distributed LLL were based on randomized algorithms.
The bounds given in the following hold w.h.p.
The bound of $O(\log^2 n)$ for the algorithm by Moser and Tardos \cite{MoserTardos10} is due to $O(\log n)$ steps in which variables are resampled, where in each step a maximal independent set (MIS) is computed in $O(\log n)$ rounds in order to perform the resampling in a conflict-free manner.
By showing that a weaker variant of an MIS is sufficient for this purpose, Chung, Pettie and Su \cite{SuLLL2017} obtained an upper bound of $O(\log^2 d \log n)$, for the same LLL criterion $ep(d+1) < 1$.
In turn, by improving the computation of such a weak MIS from $O(\log^2 d)$ to $O(\log d)$, Ghaffari \cite{GhaffariImproved16} improved this bound to $O(\log d \log n)$.

In the aforementioned work, Chung, Pettie and Su also showed that faster algorithms can be obtained if the LLL criterion is strengthened: under the criterion $epd^2 < 1$, they provide an algorithm running in time $O(\log n)$, and under an exponential criterion, i.e., a criterion of the form $pf(d) < 1$ where $f(d)$ is exponential in $d$, they give an upper bound of $O(\log n / \log \log n)$.
For LLL instances with $d \in O(\log^{1/5} \log n)$, Fischer and Ghaffari \cite{ManuelaLLL17} provided a $2^{O(\sqrt{\log \log n})}$-round algorithm under the criterion $p(ed)^{32} < 1$.
Ghaffari, Harris and Kuhn \cite{newHypergraphMatching} improved on this result by showing that for any integer $i \geq 1$, there is an LLL algorithm running in time $\operatorname{exp}^{(i)}(O(\log d + \sqrt{\log^{(i+1)} n}))$ under the criterion $20000pd^8 \leq 1$, where $\operatorname{exp}^{(i)}$ and $\log^{(i)}$ represent a power tower and the iterated logarithm, respectively.
Finally, Rozhon and Ghaffari \cite{Rozhon19} proved, as one of the many implications of their recent breakthrough in computing network decompositions, that on bounded-degree graphs a variable assignment avoiding all events can be found in $O(\poly \log \log n)$ rounds under the criterion $pd^{10} < 1$, closing in on a conjecture by Chang and Pettie \cite{ChangHierarchy17} stating that $O(\log \log n)$ rounds are sufficient.

The latter three works \cite{ManuelaLLL17,newHypergraphMatching,Rozhon19} also provide the first non-trivial \emph{deterministic} algorithms for the distributed LLL. The currently best known upper bound by Rozhon and Ghaffari \cite{Rozhon19} (for polynomial criteria) states that $\poly \log n$ rounds suffice under the criterion $epd(1+\eps) < 1$, for any constant $\eps > 0$.

\section{Preliminaries}
\subsection{Model}\label{sec:model}
The model in which we study the LLL is the LOCAL model of distributed computing \cite{linial92,peleg00}.
In the LOCAL model, we usually want to solve a graph problem, but unlike in centralized computation, the actual computation is performed by the nodes of the input graph.
To this end, each node of the input graph is considered as a computational entity, and each edge as a communication link over which the entities can communicate.
The computation proceeds in synchronous rounds, where in each round two things happen: first, each node sends an arbitrarily large message to each of its neighbors and then, after the messages have arrived, each node can perform an arbitrarily complex internal computation.
Each node has to decide at some point that it terminates and then it must output its local part of the global solution to the given problem---in the case of the LLL problem this local part is the values of all random variables the associated event depends on.
The runtime of a distributed algorithm is the number of rounds until the last node terminates.

\subsection{The Reduction}\label{sec:prevalgo}
In this section, we will give a detailed explanation of the argumentation presented in \cite{brandt2019sharp} that reduces proving the existence of an $O(d^2 + \log^* n)$-round distributed deterministic LLL algorithm under the criterion $p2^d < 1$ to showing that a certain family of sets or functions is convex.
The blueprint for such an algorithm $\alg$ is given as follows.

Consider an instance of the LLL, given by a set $\{ X_1, \dots, X_m \}$ of mutually independent random variables and a set of events that depend on the random variables.
Consider the dependency graph $G = (V, E)$ of this instance, and denote the event associated with a vertex $v$ by $\evt_v$, and the maximum degree of $G$ by $d$.
Let $p$ be a parameter such that each event occurs with probability at most $p$, and assume that $p2^d < 1$, i.e., fix the LLL criterion to $p2^d < 1$.
As any two events that depend on the same variable are neighbors of each other in $G$, we can create for each random variable $X_i$ a hyperedge that has the nodes $v$ such that $\evt_v$ depends on $X_i$ as endpoints.
Technically, the hyperedges are not part of $G$, but for simplicity, we might consider them as such.

Algorithm $\alg$ starts by computing a $2$-hop coloring with $O(d^2)$ colors in $\tilde{O}(d) + O(\log^* n)$ rounds, by applying the coloring algorithm by Fraigniaud, Heinrich and Kosowski \cite{fraigniaud16} to $G^2$, i.e., to the graph obtained by connecting any two nodes of distance at most $2$ in $G$ by an edge.
Then, it iterates through the colors one by one, and each time a color $c$ is processed, each node $v$ of color $c$ fixes each incident random variable (i.e., each random variable whose corresponding hyperedge is incident to $v$) that has not been fixed so far.
We will see that in order to fix all incident random variables of a node in a suitable way, $O(1)$ rounds suffice, and as there are $O(d^2)$ colors, algorithm $\alg$ runs in $O(d^2 + \log^* n)$ rounds.

The challenging part is to fix the random variables in a manner such that the produced final assignment is correct, i.e., such that none of the events occurs under the assignment.
To this end, during the fixing process the authors keep track of, roughly speaking, how favorable or unfavorable the variable fixings performed so far were for the nodes (regarding avoiding the associated event), by assigning two values to each edge.
More precisely, they assign a non-negative value $\varphi_e^v$ to each pair $(e, v) \in E \times V$ for which $e$ is incident to $v$.
We can imagine the two values $\varphi_e^u$ and $\varphi_e^v$ to be written on edge $e$; each time a random variable $X_i$ is fixed by a node, the node also updates the values that are written on the edges in the skeleton of the hyperedge corresponding to $X_i$.

The purpose of these edge values w.r.t. obtaining a correct output in the end of the process is to define a property $P^*$ that is kept as an invariant during the fixing process and guarantees that the final assignment avoids all events.
Consider an arbitrary point in the fixing process where some random variables $X_1, \dots, X_\ell$ already have been fixed to some values $x_1, \dots, x_\ell$, respectively.
Property $P^*$ is satisfied if the following two conditions hold.
\begin{enumerate}
    \item $\varphi_e^u + \varphi_e^v \leq 2$ for each edge $e = \{ u, v \}$.
    \item $\prob{\evt_v \mid X_1 = x_1, \dots, X_\ell = x_\ell} \leq p \cdot \prod_{e \ni v} \varphi_e^v $ for each node $v$.
\end{enumerate}
If Property $P^*$ is satisfied when all variables have been fixed, then for each event $\evt_v$ we have a bound of $p \cdot \prod_{e \ni v} \varphi_e^v \leq p2^d < 1$ for the probability that $\evt_v$ occurs, which implies that $\evt_v$ does not occur since the probability of it occurring can only be $0$ or $1$.
By initializing each value $\varphi_e^v$ to $1$, the authors make sure that $P^*$ is satisfied when algorithm $\alg$ starts.
The crucial insight in \cite{brandt2019sharp} is that there is always a way to preserve Property $P^*$ each time a random variable is fixed if a certain function or set is convex.
For the precise statement, the authors introduce the notion of a representable triple.
\begin{definition}[Definition 3.3 of \cite{brandt2019sharp}]
    A triple $(a,b,c) \in \R_{\geq 0}^3$ is called \emph{representable} if there are values $a_1, a_2, b_1, b_3, c_2, c_3 \in [0,2]$ such that $a_1 \cdot a_2 = a$, $b_1 \cdot b_3 = b$, $c_2 \cdot c_3 = c$, $a_1 + b_1 \leq 2$, $a_2 + c_2 \leq 2$, and $b_3 + c_3 \leq 2$.
Let $S_{\operatorname{rep}}=\{(a,b,c)\in \R_{\geq 0}^3 \mid (a,b,c)\text{ is representable}\}$ denote the set of all representable triples.
\end{definition}
Using this definition, the authors prove the following statement for the case of rank-$3$ random variables (which we give in a reformulated version using the notion of convexity instead of the concept of ``incurvedness" used in \cite{brandt2019sharp}).

If $[0,2]^3 \setminus S_{\operatorname{rep}}$ is a convex set, then there is a way to fix any given random variable $X_i$ of rank at most $3$ at any point in time during the algorithm (or, more generally, for any arbitrary fixing of already fixed random variables such that Property $P^*$ is satisfied) such that Property $P^*$ is preserved.
Moreover, the only information required to fix $X_i$ is the set of values $\varphi_e^v$ written on the edges $e$ that belong to the skeleton of the hyperedge corresponding to $X_i$.
We refer to \cite[Section 3.3]{brandt2019sharp} for the details of the proof.

Hence, in algorithm $\alg$, each node $v$ that has the task to fix all its incident unfixed random variables can simply collect all edge values written on edges between nodes in its inclusive $1$-hop neighborhood, and then go through its incident random variables one by one, each time finding a value for the random variable in question that preserves Property $P^*$.
As the sequential fixing does not require any communication after obtaining the required edge values, fixing \emph{all} incident unfixed variables of a node can be done in $O(1)$ rounds.
Moreover the local nature of $P^*$ and the fact that the set of edge values required and rewritten by a node during the fixing does not intersect with the set of analogous edge values for a node in distance at least $3$ ensures that any two nodes with the same color in the computed $2$-hop coloring can perform the variable fixing in parallel.
This concludes the description of the reduction.

As already noted by the authors, the definitions and proofs (for the reduction to the convexity statement) generalize straightforwardly to the case of random variables of arbitrary rank. However, showing that the convexity of the respective set indeed holds for higher dimensions remained unanswered in \cite{brandt2019sharp}; and indeed, even given our resolution, it remains unclear and would be interesting to see whether their analytical approach can feasibly be extended to higher dimensions than $3$. 
To be precise, their approach extends in the following way: to prove the existence of the deterministic algorithm in the case that each random variable affects at most $r$ events, it suffices to prove that the set $ \snon^{(r)} := [0,1]^r \setminus \srep^{(r)}$ is convex, where $\srep^{(r)}$ is the set of all representable tuples, which are tuples that can be generated by some generator, as defined below. 
\begin{definition}[generator]
\label{def:generator}
We call a vector $(a_{ij})_{i \neq j \in [r]}$ with $r(r-1)$ coordinates a \emph{generator} if for each $i\not=j$ we have $0 \leq a_{ij} \leq 1$ and $a_{ij} + a_{ji} \le 1$. The generator  $(a_{ij})_{i \neq j \in [r]}$ generates the $r$-dimensional tuple $(a_1, \ldots, a_r)$ with $a_i = \prod_{j \in [r] \setminus \{i\}}a_{ij}$ for $i \in [r]$.
We call a generator \emph{non-zero}, if none of its coordinates is $0$. 
We use a shorthand notation and denote the generator $(a_{ij})_{i \neq j \in [r]}$ simply as $(a_{ij})$.
\end{definition}
Note that if $(a_{ij})$ is a non-zero generator, then $a_{ij} < 1$ for each $i \neq j \in [r]$.

\begin{definition}[representable tuples]
A tuple $(a_1, \ldots, a_r) \in \mathbb{R}^r_{\geq 0}$ is called representable if there exists a generator $(a_{ij})$ that generates it. Let $\srep^{(r)} = \{(a_1, \ldots, a_r) \in \mathbb{R}^r_{\geq 0}|(a_1, \ldots, a_r) \text{ is representable }\}$ denote the set of all representable tuples.
\end{definition}
Note that $\srep^{(3)} \neq \srep$, as we require $a_{ij} + a_{ji} \leq 1$ instead of $a_{ij} + a_{ji} \leq 2$. We consider this scaled version, as this makes the proof cleaner later on: note that $[0,1]^3 \setminus \srep^{(3)}$ being convex directly implies that $[0,2]^3 \setminus \srep$ is convex as the latter is just a scaled variant of the former set. 
In the following, we drop the superscripts when clear from context and we denote with $\srep$ the set of representable tuples with respect to the scaled down version and $\snon$ as the set of points in $[0,1]^r$ which are not representable.
Our main contribution is the proof of the following theorem. 
\begin{theorem}
\label{thm:main}
For every $r \geq 2$, $\snon^{(r)}$ is convex.
\end{theorem}
This settles \cref{conj:bmu} as described above. 

\section{Proving that $\snon$ is convex}\label{sec:proof}

In this section we prove that set $\snon$ is convex, omitting two longer proofs that are postponed to \cref{sec:technical} and \cref{sec:hyperplane}. 

\subsection{Notation}
We work with the standard Euclidean space $\mathbb{R}^m$ where distances are measured with the Euclidean norm; $\textbf{0}$ and $\textbf{1}$ denote the vectors $(0,0,\dots,0)^T$ and $(1,1,\dots,1)^T$, respectively. 
We define $B(x,R) := \{ y \in \mathbb{R}^m, \Vert x - y \Vert \le R\}$ as the closed ball around $x$ with radius $R$. A subset $S \subseteq \mathbb{R}^m$ is open if for any $x \in S$, there exists $R > 0$ such that $B(x,R) \subseteq S$. A subset $S \subseteq \mathbb{R}^m$ is closed if $\mathbb{R}^m \setminus S$ is open. A set $S \subseteq \mathbb{R}^m$ is bounded if there exists $R >0$ such that $S \subseteq B(\mathbf{0}, R)$. A set $S$ is compact if it is closed and bounded. Equivalently, $S$ is compact if every sequence $x_1, x_2, \dots $ with each $x_i \in S$ has a subsequence $x_{s(i)}$ that converges to some $x \in S$. The subset $[0,1]^m \subseteq \mathbb{R}^m$ is compact. 
The interior of a set $S$ is an open subset of $S$ and defined as $S^\textrm{o} = \{ x\in S, \exists R>0: B(x,R) \subseteq S\}$. 
The boundary of a set $S$ is defined as $\partial S = \{ x \in \mathbb{R}^m,  \forall R > 0 : B(x,R)\cap S \not= \emptyset \text{ and } B(x,R)\cap (\mathbb{R}^m\setminus S) \not= \emptyset \}$. 
A set $S$ is path-connected if for any $x,y \in S$ there exists a continuous function $f: [0,1] \rightarrow S$ such that $f(0) = x$ and $f(1) = y$. 

A hyperplane $H \subset \mathbb{R}^m$ is an affine subspace of dimension $m-1$. Equivalently, it is a set of points $H = \{x\in \mathbb{R}^m, h^T x = b\}$ for some vector $h \in \mathbb{R}^m \setminus \{\mathbf{0}\}$ and $b \in \mathbb{R}$.
A weakly supporting hyperplane for $S$ intersecting $y \in \partial S$ is a hyperplane $H = \{x\in \mathbb{R}^m, h^T x = b\}$ with $h^T y = b$ and $h^T z \geq b$ for any $z \in S$. 
Finally, a weakly locally supporting hyperplane for $S$ intersecting $y \in \partial S$ is a hyperplane $H = \{x\in \mathbb{R}^m, h^T x = b\}$ with $h^Ty = b$ satisfying the following property: there exists an $\eps > 0$ such that for any $z \in S \cap B(y, \eps)$ we have $h^T z \ge b$.

\subsection{Proof}

Convexity of a set can be verified in several equivalent ways. 
As we outlined in \cref{subsec:techniques}, we rely on the ``supporting hyperplane formulation'', i.e., a set is convex if for each boundary point we can find a hyperplane such that the whole set lies on one side of the hyperplane. 
Moreover, for connected sets, it is enough to prove that each such hyperplane is ``locally'' supporting as formalized in the following theorem, which is stated in a more general form in  \cite{valentine} (Theorem 4.10 there). 

\begin{theorem}
\label{theorem:book}
Let $S \subseteq \mathbb{R}^r$ be an open and path-connected set in $\mathbb{R}^r$. 
The set $S \subseteq \mathbb{R}^r$ is convex if for every point $y$ contained in the boundary of $S$, there exists a weakly locally supporting hyperplane with respect to $S$ going through $y$. 
\end{theorem}

Note that \cref{theorem:book} can only be used to prove convexity of open sets and thus cannot directly applied to establish the convexity of $\snon$. Instead, we use \cref{theorem:book}  to first establish convexity of the interior of $\snon$, which is an open set and which we denote by $\isnon $. Once we have established the convexity of $\isnon $, we prove the convexity of $\snon$ by induction on the dimension $r$. To prove convexity of $\isnon $, we need to show that $\isnon $ is path-connected and that for every boundary point of $\isnon $, there exists a weakly locally supporting hyperplane going through the boundary point. We now prove the former, using the following simple observation, which will be used in several other proofs. 

\begin{observation}
\label[observation]{observation:below_maximal_is_representable}
Let $a = (a_1, \dots, a_r)$ be a representable tuple. Then any tuple $a'$ with $0 \le a'_i \le a_i$ for all $i \in [r]$ is also representable. 
\end{observation}
\begin{proof}
Consider a generator $(a_{ij})$ of $a$. 
For any $i$, pick some $j\not=i$ and set $a'_{ij} = a_{ij} \cdot \frac{a'_i}{a_i} \leq 1$. Set all other values in $(a'_{ij})$ equal to the corresponding value in $(a_{ij})$. $(a'_{ij})$ is a valid generator generating the tuple $a'$. 
\end{proof}

Now, we are ready to prove that $\isnon $ is path-connected. 

\begin{lemma}
The set $\isnon $ is path-connected. 
\end{lemma}
\begin{proof}
For any $u, u' \in \isnon $, consider the vector $u'' \in \mathbb{R}^r$ with $u''_i = \max\{u_i, u_i'\} > 0$ for every $i \in [r]$. Note that the union of the two segments between $u$ and $u''$ and between $u''$ and $u'$ is a path. Moreover, any tuple on this path is contained in $(0,1)^r$ and either dominates $u$ or $u'$. Hence, by \cref{observation:below_maximal_is_representable}, each tuple on the path is in $\isnon $. 
\end{proof}

Next, we need to understand the boundary between $\srep$ and $\snon$. To do so, it will be helpful to prove that $\srep$ is closed. As $\srep \subseteq [0,1]^r$ is bounded, this is equivalent to show that $\srep$ is compact. 

\begin{lemma}
\label{lemma:compact}
The set $\srep$ is compact. 
\end{lemma}
\begin{proof}
The set $\srep$ is defined as an image of a continuous function that maps each generator  (\cref{def:generator}) from the compact set of all generators to the corresponding representable tuple. 
Hence, it is compact as an image of a compact set under continuous function is always compact. 
\end{proof}

Next, we set up the notion of maximal tuples. 
\begin{definition}[domination and maximal tuples]
Let $a = (a_1, \dots, a_r)$ and $a' = (a'_1, \dots, a'_r)$ be two representable tuples.
We say that $a'$ \emph{weakly dominates} $a$ if $a'_i \ge a_i$ for all $i \in [r]$, and $a' \not = a$.
Moreover, we say that $a'$ \emph{strongly dominates} $a$ if $a'_i > a_i$ for all $i \in [r]$.
We call a representable tuple $a$ \emph{maximal} if there is no representable tuple $a'$ that weakly dominates $a$.
\end{definition}

Intuitively, maximal tuples are forming the boundary between $\srep$ and $\snon$ and this is indeed what we prove. 

\begin{lemma}
\label{lemma:boundary}
Let $x \in \mathbb{R}^r$ be contained in $\partial \snon $. Then, there either exists $i \in [r]$ such that $x_i \in \{0,1\}$ or $x$ is a maximal representable tuple.
\end{lemma}

We defer the easy, yet slightly technical proof, together with proofs of a few other technical lemmas, to \cref{sec:technical}. 
Our main technical contribution is a proof that a locally supporting hyperplane can be found for any maximal tuple $a$. 

\begin{lemma}
\label{lemma:maximaltuplehyperplane}
For each maximal representable tuple $a$, there exists a locally supporting hyperplane for $\isnon$ intersecting $a$. 
\end{lemma}

The non-trivial proof of the above lemma is deferred to \cref{sec:hyperplane}. As a corollary, we infer that the whole set $\isnon$ is convex. 

\begin{corollary}
\label{cor:convex_interior}
The set $\isnon$ is convex. 
\end{corollary}
\begin{proof}
By \cref{theorem:book} it suffices to provide a weakly locally supporting hyperplane for any $a \in \partial \snon$. 
By \cref{lemma:boundary}, any $a \in \partial \snon$ is either a maximal representable tuple and hence the existence of the supporting hyperplane follows from \cref{lemma:maximaltuplehyperplane}, or we have $a_i = 0$ or $a_i = 1$, respectively, for some $i$. 
But then the hyperplane $\{x \in \mathbb{R}^r \colon e_i^Tx = 0\}$ or $\{x \in \mathbb{R}^r \colon -e_i^Tx = -1\}$, respectively, is a weakly (locally) supporting hyperplane for $\isnon$ intersecting $a$. 
\end{proof}


The proof of \cref{thm:main} now easily follows.

\begin{proof}[Proof of \cref{thm:main}]
We prove the statement by induction on $r$. For $r = 2$, the statement trivially holds. Now, let $r \geq 3$ arbitrary and assume that $\snon^{(r-1)}$ is convex. Let $x \neq y \in \snon^{(r)}$ and $\alpha \in (0,1)$ be arbitrary. We need to show that for $z:= \alpha x + (1-\alpha)y$ we have $z \in \snon^{(r)}$. 
As $\srep^{(r)}$ is a closed set  (\cref{lemma:compact}), there exists some $\varepsilon$ with $0 < \eps < \min(\alpha, 1 - \alpha)$ such that $x' = (1 - \eps)x + \eps y \not\in \srep^{(r)}$ and, hence, $x' \in \snon^{(r)}$ since the whole segment $\{\beta x + (1-\beta)y, 0 < \beta < 1 \}$ is contained in $[0,1]^r$, and $y':= (1-\varepsilon)y + \varepsilon x \in \snon^{(r)}$. Furthermore, there exists an $\alpha' \in (0,1)$ such that $z = \alpha'x' + (1-\alpha')y'$. 

If $x',y' \in \isnon ^{(r)}$, then, by \cref{cor:convex_interior}, it follows that $z \in \isnon ^{(r)}$ and we are done. 
Otherwise, $x' \not \in\isnon ^{(r)}$ or $y' \not \in\isnon ^{(r)}$. Without loss of generality, assume that $x'\not \in\isnon ^{(r)}$. Since $x' \not \in \snon^{(r)}$, \cref{lemma:boundary} implies that there exists some $i\in [r]$ with $x_i' \in \{0,1\}$. 
Our choice of $\eps > 0$ now implies that either $x_i=y_i=z_i=1$ or $x_i=y_i=z_i=0$. 

In the first case, as $x$ is not representable, there exists some $j \in [r] \setminus \{i\}$ with $x_j > 0$. 
Therefore, $z_j > 0$ and as $z_i = 1$, any generator of $z$ would need to have $z_{ij} = 1$ and $z_{ji} > 0$, a contradiction with $z_{ij} + z_{ji} \le 1$. Hence, $z \in \snon^{(r)}$. 

In the second case, assume without loss of generality that $i=r$. 
Let $\Tilde{x}$, $\Tilde{y}$, $\Tilde{z} \in [0,1]^{r-1}$ be equal to the vectors $x$, $y$ and $z$ restricted to the first $r-1$ coordinates. We have $\Tilde{x},\Tilde{y} \in \snon^{(r-1)}$, since otherwise taking their generator and augmenting it by zeros would generate $x$ or $y$, respectively. 
As $\Tilde{z}$ is a convex combination of $\Tilde{x}$ and $\Tilde{y}$, the induction hypothesis implies that $\Tilde{z} \in \snon^{(r-1)}$ and therefore $z \in \snon^{(r)}$, which concludes the induction step. 
\end{proof}

\section{Technical preparation}
\label{sec:technical}

In this section we prove several technical results that are needed for the proof. 
First, we prove the equivalence of the notions of weak and strong dominance. 
To this end, we first show a simple ``continuity'' statement that shows that for any representable tuple $a$, one can increase all but one of its coordinates a little bit at the expense of decreasing the remaining one. 
\begin{lemma}
\label{lemma:tradingoffepsilons}
Let $(a_1, \dots, a_r)$ be a representable tuple with $a_i > 0$ for each $i \in [r]$. For each $k \in [r]$, there exist an $\eps > 0$ and a $\xi > 0$ such that for all $t$ with $0 < t < \eps$, the tuple $a'$ defined by $a'_k = a_k - t$ and $a'_i = a_i + \xi t$ for $i \not = k$ is also representable. 
\end{lemma}
\begin{proof}
Let $(a_{ij})$ be a generator of $(a_1, \ldots, a_r)$. As $a_i > 0$ for each $i \in [r]$, $(a_{ij})$ is a non-zero generator. Now, for some $\delta > 0$, consider $(b_{ij})$ with $$b_{ij} = \begin{cases}
a_{ij} - \delta \text{ if $i = k$} \\
a_{ij} + \delta \text{ if $j = k$} \\
a_{ij} \text{\hspace{15pt} otherwise }
\end{cases}$$
for each $i \neq j \in [r]$. We have $b_{ij} + b_{ji} = a_{ij} + a_{ji} \leq 1$ for each $i, j \in [r], i\not= j$. Furthermore, if we choose $\delta$ such that $0 < \delta < \varepsilon' := \min_{i \neq j \in [r]} \min(a_{ij},1-a_{ij}) < 1$, we have $0 \leq b_{ij} \leq 1$ for each $i,j \in [r], i\not = j$. 
In that case, $(b_{ij})$ is a valid generator that generates a tuple $(b_1, \ldots, b_r)$ such that:

\[b_k = \prod_{j \neq k} b_{kj} 
= \prod_{j \neq k} (a_{kj} - \delta) 
\geq \left( \prod_{j \neq k} a_{kj} \right) - \delta f((a_{ij})) = a_k - \delta f((a_{ij})) \]

for some function $f$ with $f((a_{ij})) > 0$. Note that such a function $f$ exists, as $\delta < 1$ and therefore $\delta^e \leq \delta$ for each $e \geq 1$. For each $i \in [r] \setminus \{k\}$, we have:

\begin{align*}
b_i &= \prod_{j \neq i} b_{ij} 
= (a_{ik} + \delta) \cdot \prod_{j \notin \{i,k\}} a_{ij} = a_i + \delta \cdot  \prod_{j \notin \{i,k\}} a_{ij}\\
&\geq  a_i + \delta \cdot  \prod_{j \neq i} a_{ij} = a_i + \delta a_i
\end{align*}

Set $t = \delta \cdot f((a_{ij}))$, $\xi = \frac{1}{f((a_{ij}))}\min_{i \in [r]} a_i > 0$ and $\varepsilon = \varepsilon' \cdot f((a_{ij})) > 0$.  Now, consider some arbitrary $t$ with $0 < t < \varepsilon$. 
The definition of $\varepsilon$ implies that $0 < \delta = \frac{t}{f((a_{ij}))} < \frac{\varepsilon}{f((a_{ij}))} = \varepsilon'$. Thus, we can represent a tuple $(b_1, \ldots, b_r)$ with $b_k \geq a_k - \delta f((a_{ij})) \geq a_k - t = a'_k$ and $b_i \geq a_i +  \delta \cdot  a_i \geq a_i + \xi \cdot t = a'_i$ for $i \neq k$. This tuple dominates the tuple $a'$. As $a'_i \geq 0$ for each $i \in [r]$, \cref{observation:below_maximal_is_representable} implies that we can represent $a'$.
\end{proof}

Now we are ready to show that if a tuple is weakly dominated by some other tuple, it is also strongly dominated by (a potentially different) one. 

\begin{corollary}[strong vs weak domination]
\label{cor:strong_domination}
Let $a = (a_1, \dots, a_r)$ be a representable tuple such that for all $i \in [r]$ we have $0 < a_i < 1$. If there exists a representable tuple that weakly dominates $a$, then there also exists a representable tuple that strongly dominates $a$. 
\end{corollary}

\begin{proof}
Let $(a'_1,\ldots,a'_r)$ be a representable tuple that weakly dominates $a$. Note that we have $a'_i > 0$ for each $i \in [r]$. Let $k \in [r]$ such that $a'_k > a_k$. According to $\cref{lemma:tradingoffepsilons}$, there exists $\varepsilon > 0$ and $\xi > 0$ such that for all $0 < t < \varepsilon$, the tuple $(a'_1 + \xi t,  \ldots, a'_{k-1} + \xi t, a'_k - t, a'_{k+1} + \xi t, \ldots, a'_r + \xi t)$ is representable. For $t$ small enough, this tuple strictly dominates the tuple $(a_1, \ldots, a_r)$.
\end{proof}

We are now ready to prove \cref{lemma:boundary}.

\begin{proof}[Proof of $\cref{lemma:boundary}$]
We show the contrapositive. Let $x \in \mathbb{R}^r$ such that $x_i \notin \{0,1\}$ for each $i \in [r]$ and $x$ is not a maximal representable tuple. We show that this implies the existence of a ball $B(x,\eps)$ around $x$ with radius $\eps > 0$ such that either $ B(x,\eps) \subseteq \snon $ or $B(x,\eps) \cap \snon = \emptyset$, which in turn implies $x \not\in \partial \snon$. 

For $x \notin [0,1]^r$ there is clearly such a ball. 
Otherwise, we have $x \in (0,1)^r$, as we assume that for each $i \in [r]$, $x_i \notin \{0,1\}$.
If $x \in \srep$, but $x$ is not maximal representable, there is a representable tuple that weakly dominates $x$ and as $x \in (0,1)^r$, \cref{cor:strong_domination} provides a representable tuple that strongly dominates $x$. 
Hence, there exists some $\varepsilon > 0$ such that the tuple $(x_1 + \varepsilon, \ldots, x_r + \varepsilon)$ is a representable tuple. 
This implies that for $\eps' = \min(\varepsilon, \min_{i \in [r]} x_i)$ we have $B(x, \eps') \subseteq \srep$ due to \cref{observation:below_maximal_is_representable} and, hence, $B(x, \eps') \cap \snon = \emptyset$ as needed. 
Finally, in the case $x \in (0,1)^r \cap \snon$ \cref{lemma:compact} implies that the complement of $\srep$ is open which in turn implies the existence of an $\eps > 0$ so that $B(x,\eps) \cap \srep = \emptyset$. 
For $\eps' = \min(\varepsilon, \min_{i \in [r]} \min(x_i,1-x_i))$ we then have $B(x, \eps') \subseteq [0,1]^r \setminus \srep = \snon$, as needed.  
\end{proof}

\section{Construction of hyperplanes}
\label{sec:hyperplane}

In this section we prove our main technical contribution: \cref{lemma:maximaltuplehyperplane} that states that for each maximal tuple we can find a locally weakly supporting hyperplane for the set $\snon$. 
First, we give an informal proof of this result for the case $r=3$, which captures the intuition behind the general proof for all $r$ that we give later. 

\subsection{Informal outline for $r=3$}
\label{subsec:r=3}

Our main observation is that finding a locally supporting hyperplane comes down to proving that a certain set of tuples in the neighbourhood of $a$ is representable. 
In this section, we denote the tuples, more intuitively, as triples. 
So, we now focus on how to generate triples similar to $a$. 
	
\paragraph{Generating more triples}
Given a representable triple $a \in (0,1)^3$ generated by the generator $(a_{ij})$, what other triples close to $a$ are representable? 
	Certainly, all the triples that $a$ dominates. 
	Besides, we can play with the generator itself. 
	Adding $\alpha_{12}$ to $a_{12}$ and subtracting it from $a_{21}$ gives us again a valid generator that generates triples of the form
	\begin{align*}
	&\left( a_{13} (a_{12} + \alpha_{12}), (a_{21} - \alpha_{12}) a_{23}, a_{31} a_{32} \right)
	= a + \alpha_{12} \left( a_{13}, -a_{23}, 0\right)
	\end{align*}
	I.e., it generates triples on the line 
	\[
	a + \alpha_{12}(a_{13}, -a_{23}, 0) = a + \alpha_{12} w_{12}, \;\; \alpha_{12} \in \mathbb{R}
	\]
	for $|\alpha_{12}|$ small enough. Similarly, by adding $\alpha_{13}$ to $a_{13}$ and subtracting it from $a_{31}$, we can generate triples on the line
	\[
	a + \alpha_{13} (a_{12}, 0, -a_{32})= a + \alpha_{13} w_{13}, \;\; \alpha_{13} \in \mathbb{R}
	\]
	and by adding $\alpha_{23}$ to $a_{23}$ and subtracting it from $a_{32}$, we can generate triples on the line
	\[
	a + \alpha_{23} (0, a_{21}, -a_{31})= a + \alpha_{23} w_{23}, \;\; \alpha_{23} \in \mathbb{R} 
	\]
	in some neighborhood around the triple  $a$. We call the three lines $\ell_1, \ell_2$ and $\ell_3$. 
	
	Since all components of the generator of $a$ are nonzero, these three lines define an affine subspace of dimension at least two. 
	Later we prove that if $a$ is a maximal representable triple, then the three lines lie on a common plane. 
	The plane spanned by $\ell_1, \ell_2$ and  $\ell_3$ then becomes an obvious suspect for the supporting hyperplane we wish to find!
	
	In fact, we prove that not only triples on the lines $\ell_1, \ell_2$ and $\ell_3$ are representable, given that they lie in some small neighborhood around the maximal representable triple $a$, but \emph{any triple $a'$ in the affine span of the three lines} is representable, provided that $a'\in B(a,\eps)$ for some positive $\eps$ that depends on $a$. 
	This finishes our proof, as we can now find a weakly locally supporting hyperplane for each maximal representable triple $a$.
	
	We now prove that for maximal triples all three lines lie in a common plane and all triples in that plane are representable (if they are close enough to $a$).

	\begin{claim}
	\label{claim:plane}
For a maximal triple $a$, the affine hull of $\ell_1, \ell_2$ and $\ell_3$ is a plane. 
	\end{claim}
    
	Assume the contrary.  
	Then, there exist $\alpha_{12},\alpha_{13},\alpha_{23} \in \mathbb{R}$, such that  $(1,1,1) = \alpha_{12} w_{12} + \alpha_{13} w_{13} + \alpha_{23} w_{23}$. Now, change the values of $(a_{ij})$ proportional to the values of $\alpha$ to obtain the generator $(a'_{ij})$ with
	\begin{align*}
	a'_{12} &= a_{12} + \xi \alpha_{12}, \qquad
	a'_{21} = a_{21} - \xi \alpha_{12};\\
	a'_{13} &= a_{13} + \xi \alpha_{13},\qquad
	a'_{31} = a_{31} - \xi \alpha_{13};\\ 
	a'_{23} &= a_{23} + \xi \alpha_{23}, \qquad
	a'_{32} = a_{32} - \xi \alpha_{23}. 
	\end{align*}

	Intuitively, we expect these changes to give us a generator of $a' = a + \xi \alpha_{12} w_{12} + \xi \alpha_{13} w_{13} +\xi \alpha_{23} w_{23} = a + \xi \cdot (1,1,1)$. 
	This is almost the case:
	\begin{align*}
	a'&= \left( a'_{12}a'_{13}, a'_{21}a'_{23}, a'_{31}a'_{32}\right)\\
	&= ( (a_{12} +  \xi \alpha_{12})(a_{13} + \xi \alpha_{13}), (a_{21} - \xi \alpha_{12})(a_{23} + \xi \alpha_{23}),\\ &\quad\;\; (a_{31} - \xi \alpha_{13})(a_{32} - \xi \alpha_{23}))\\
	&= \left( a_{12}a_{13}, a_{21}a_{23}, a_{31}a_{32}\right) + \xi \alpha_{12}(a_{13}, -a_{23},0)\\ &\quad + \xi \alpha_{13}(a_{12}, 0,-a_{32}) + \xi \alpha_{23}(0,a_{21},-a_{31})\\ &\quad + \xi^2(\alpha_{12}\alpha_{13}, - \alpha_{12}\alpha_{23}, \alpha_{13}\alpha_{23}) \\
	&= a + \xi\alpha_{12}w_{12} + \xi\alpha_{13}w_{13} + \xi\alpha_{23}w_{23}\\ &\quad + \xi^2(\alpha_{12}\alpha_{13}, - \alpha_{12}\alpha_{23}, \alpha_{13}\alpha_{23})\\ 
	&= a + \xi (1,1,1) +  \xi^2(\alpha_{12}\alpha_{13}, - \alpha_{12}\alpha_{23}, \alpha_{13}\alpha_{23})
	\end{align*}
	
	Choosing $\xi > 0$ small enough, we conclude that the triple $a + \xi/2 \cdot (1,1,1)$ is representable and therefore $a$ is not maximal, a contradiction!
	
	\paragraph{Generating triples on the plane}
	We are given a maximal triple $a$ and some $a'$ in the affine hull of $\ell_1, \ell_2$ and $\ell_3$ that is sufficiently close to $a$. 
	We need to prove that $a'$ is representable. To do so, we first note that as $a'$ is contained in the affine hull, there exist $\alpha_{12},\alpha_{13}$ and $\alpha_{23}$ such that $a' = a + \alpha_{12}w_{12} + \alpha_{13}w_{13} + \alpha_{23}w_{23}$. Now, we employ the same strategy as above and observe that we can change the generator of $a$ as follows
	\begin{align*}
	a'_{12} = a_{12} + \alpha_{12}, \qquad
	a'_{21} = a_{21} - \alpha_{12};\\
	a'_{13} = a_{13} + \alpha_{13},\qquad
	a'_{31} = a_{31} - \alpha_{13};\\
    a'_{23} = a_{23} + \alpha_{23}, \qquad
	a'_{32} = a_{32} - \alpha_{23}, 
	\end{align*}

	so as to generate the triple
	\begin{align*}
	&( (a_{12} +  \alpha_{12})(a_{13} + \alpha_{13}), (a_{21} - \alpha_{12})(a_{23} + \alpha_{23}),\\ &\;\; (a_{31} -  \alpha_{13})(a_{32} -  \alpha_{23}))\\
	&= \left( a_{12}a_{13}, a_{21}a_{23}, a_{31}a_{32}\right) +  \alpha_{12}(a_{13}, -a_{23},0) +  \alpha_{13}(a_{12}, 0,-a_{32})\\ &\quad +  \alpha_{23}(0,a_{21},-a_{31}) + (\alpha_{12}\alpha_{13}, - \alpha_{12}\alpha_{23}, \alpha_{13}\alpha_{23}) \\
	&= a + \alpha_{12}w_{12} + \alpha_{13}w_{13} + \alpha_{23}w_{23} + (\alpha_{12}\alpha_{13}, - \alpha_{12}\alpha_{23}, \alpha_{13}\alpha_{23})\\
	&= a' + (\alpha_{12}\alpha_{13}, - \alpha_{12}\alpha_{23}, \alpha_{13}\alpha_{23}) 
	\end{align*}
	The  term $(\alpha_{12}\alpha_{13}, - \alpha_{12}\alpha_{23}, \alpha_{13}\alpha_{23})$  is important now. 
	We require $(\alpha_{12}\alpha_{13}, - \alpha_{12}\alpha_{23}, \alpha_{13}\alpha_{23}) \geq \mathbf{0}$ to prove that $a'$ is a representable triple. 
	This property does not hold for every choice of the coefficients $\alpha_{12},\alpha_{13},\alpha_{23}$. 
	However, as $w_{12}$, $w_{13}$ and $w_{23}$ are linearly dependent, we have a certain flexibility to choose the $\alpha$'s. 
	In particular, one can choose the $\alpha$'s in such a way that at most $2$ of them are non-zero.
	
	It turns out that it is indeed always possible to choose the $\alpha$'s in such a way that $(\alpha_{12}\alpha_{13}, - \alpha_{12}\alpha_{23}, \alpha_{13}\alpha_{23}) \geq \mathbf{0}$. 
	To see why, observe that the first coordinate of the  quadratic term is negative if and only if out of the  two numbers $a_{12}$ and $a_{13}$ (recall that $a_1 = a_{12}a_{13}$), one is increased and one is decreased.
	This holds analogously also for the other coordinates. 
	So, we show how to generate $a'$ such that this does not happen. 
	
First, one can observe that $a'$ does not dominate $a$ or vice versa: if this was the case, one could obtain a contradiction by showing that $a$ is not a maximal representable triple similarly to the proof of \cref{claim:plane}. 
	Thus, we can assume that there exist $i \neq j \in [3]$ such that $a_i < a'_i$ and $a_j > a'_j$. Assume (without loss of generality) that $a_1 <  a_1'$, $a_2 > a_2'$ and $a_3 \ge a_3'$. 
	In that case, we first fix $\alpha_{23} = 0$. 
	As $a_2 \geq a'_2$ and $a_3 \geq a'_3$, one can show that $a'$ lies in the span of $\ell_1$ and $\ell_2$, thus one can write
	\begin{align*}
	a' &= a + \alpha_{12} w_{12} + \alpha_{13}w_{13}\\
	&= a + \alpha_{12} (a_{13}, -a_{23}, 0) + \alpha_{13}(a_{12}, 0, -a_{32})\\
	&= (a_1 + \alpha_{12}a_{13} + \alpha_{13}a_{12}, a_2-\alpha_{12}a_{23}, a_3 - \alpha_{13}a_{32}).  
	\end{align*}
	Since $a_2' < a_2$, we have $\alpha_{12} > 0$. 
	Similarly, since $a_3' \le a_3$, we have $\alpha_{13} \ge 0$. Together with $\alpha_{23} = 0$, we get $(\alpha_{12}\alpha_{13}, - \alpha_{12}\alpha_{23}, \alpha_{13}\alpha_{23}) \geq \mathbf{0}$, as needed. 
	
	This concludes the proof outline for $r=3$. 
	For general $r$, the last step is slightly more tricky: generally, we set $\alpha_{ij} = 0$ if both $a_i$ and $a_j$ needs to be increased or both needs to be increased. 
	Additionally, if $a_k$ needs to be increased, $\alpha_{ij}(w_{ij})_k$ is non-negative for all $i < j$ and if $a_{k}$ needs to be decreased, $\alpha_{ij}(w_{ij})_k$ is non-positive for all $i < j$. 
	Moreover, for general $r$, augmenting the generator according to the $\alpha$-values might lead to negative higher order terms. However, these are always dominated by the quadratic increase in a neighborhood around $a$. 
	




\subsection{Construction of hyperplanes, in general}

We start by defining ``movement vectors'', analogues to vectors $w_{12}, w_{13}, w_{23}$ from \cref{subsec:r=3}, that correspond to ``allowed movements'' that we may make to construct representable tuples in the vicinity of a representable tuple $a$. 

\begin{definition}[movement vectors]
	Let $a \in (0,1)^r$ be a maximal tuple and $(a_{ij})_{i \neq j \in [r]}$ an arbitrary (nonzero) generator of $a$.
For each $i \not= j \in [r]$, we define $w_{ij}$ as the $r$-dimensional vector such that for each $k \in [r]$, 
\[
(w_{ij})_k = \begin{cases}
\frac{a_i}{a_{ij}} \text{ if }$k = i$ \enspace, \\
\frac{-a_j}{a_{ji}} \text{ if } $k = j$ \enspace, \\
0 \text{ otherwise} \enspace.
\end{cases}
\]
\end{definition}

Similarly to \cref{subsec:r=3}, we now define the span of the movement vectors $H_a$ that we later prove to be a hyperplane for the case of maximal representable tuples. 
\begin{definition}
 We define $H_a := \{a + \sum_{i \neq j \in [r]}\alpha_{ij}w_{ij} \mid \alpha_{ij} \in \mathbb{R} \text{ for all } i \neq j \in [r] \}$. 
\end{definition}
\begin{observation}
\label{obs:dim}
$H_a$ is an affine subspace with a dimension at least $r-1$.
\end{observation}
\begin{proof}
	Consider the $r-1$ vectors $w_{12}, w_{13}, \ldots, w_{1r}$. Among those $r-1$ vectors, $w_{1j}$ is the only vector with a non-zero $j$-th coordinate. Hence, the $r-1$ vectors are linearly independent.
\end{proof}

The next lemma corresponds to  \cref{claim:plane} of the informal outline.

\begin{lemma}
\label{lemma:notinH}
Let $a \in (0,1)^r$ be a maximal tuple and $q \in \mathbb{R}_{\geq 0}$ be a non-negative vector such that there exists some index $k \in [r]$ with $q_k > 0$. Then, $a + q \notin H_a$.
\end{lemma}
\begin{proof}
We show that the existence of such a vector $q$ would contradict the fact that $a$ is a maximal tuple. Thus, for the sake of contradiction, assume that there exists a non-negative vector $q$ and some $k \in [r]$ such that $q_k > 0$ and $a + q \in H$. Thus, there exist $\alpha_{ij}$'s such that $q = \sum_{i \neq j \in [r]} \alpha_{ij}w_{ij}$. For $\delta > 0$, consider $(a'_{ij})$ with $a'_{ij} = a_{ij} + \delta \cdot (\alpha_{ij} - \alpha_{ji})$ for $i \neq j \in [r]$. Note that 
\[a'_{ij} + a'_{ji} = (a_{ij} + \delta \cdot (\alpha_{ij}- \alpha_{ji})) + (a_{ji} + \delta \cdot (\alpha_{ji} - \alpha_{ij}))  = a_{ij} + a_{ji} \leq 1. \]
Thus, for $\delta$ small enough, $(a'_{ij})$ is a valid generator. Let $a'$ denote the tuple that $(a'_{ij})$ generates. Then, for each $i \in [r]$, we get:
\begin{align*}
a'_i 
     &= \prod_{j \neq i} (a_{ij} + \delta  (\alpha_{ij} - \alpha_{ji}))\\
     &= \prod_{j \neq i} a_{ij} + \left( \sum_{\ell \neq i} \delta (\alpha_{i\ell} - \alpha_{\ell i}) \prod_{j \notin \{i,\ell\}} a_{ij} \right) - O(\delta^2)\\
     &= a_i + \left( \sum_{\ell \neq i} \delta (\alpha_{i\ell} - \alpha_{\ell i}) \frac{a_i}{a_{i\ell}} \right) - O(\delta^2) \\
     &=  a_i + \left( \sum_{\ell \neq i} \delta (\alpha_{i\ell}(w_{i\ell})_i +  \alpha_{\ell i}(w_{\ell i})_i) \right)  - O(\delta^2)\\
     &= a_i + \delta \left( \sum_{\ell \neq j} \alpha_{\ell j} (w_{\ell j})_i \right) - O(\delta^2)\\
     &= (a + \delta q)_i - O(\delta^2)
\end{align*}
Thus, there exists some constant $c \geq 0$, such that for each sufficiently small $\delta > 0$, there is a non-negative representable tuple $b$ with $b(\delta) := a + \delta q - c\delta^2 \cdot \mathbf{1}$. 
 \cref{lemma:tradingoffepsilons} implies the existence of some $\xi > 0$ such that for each sufficiently small $t > 0$, we can represent the tuple $b'(\delta,t)$ with $b'(\delta,t)_k = b(\delta) - t$ and $b'(\delta,t)_i = b(\delta) + \xi t$ for each $i \neq k$. In particular, we can choose $\delta > 0$ small enough such that for $t = (q_k/2)\delta$, the tuple $b'(\delta, t)$ is representable and furthermore:

$$b'(\delta,t)_k = b(\delta)_k - t = a_k + \delta q_k - c \delta^2 - t = a_k + \delta(q_k/2) - c \delta^2 > a_k$$
and for $i \neq k$, 
$$b'(\delta,t)_i = b(\delta)_i + \xi t = a_i + \delta q_i - c \cdot \delta^2 + \xi t \geq a_i - c \cdot \delta^2 + \xi (q_k/2)\delta > a_i$$
which contradicts the maximality of $a$. 
\end{proof}

\begin{corollary}
\label{cor:hyperplane}
For any maximal representable tuple $a$, the set $H_a$ defines a hyperplane. 
That is, there exist $h \in \mathbb{R}^r \setminus \{\bold{0}\}$ and $b \in \mathbb{R}$ such that $H_a = \{x \in \mathbb{R}^r:h^Tx = b\}$. Furthermore, one can choose $h$ such that $h \geq \mathbf{0}$.
\end{corollary}
\begin{proof}
The set $H_a$ defines an affine subspace of dimension at least $r-1$ (\cref{obs:dim}) and of dimension at most $r-1$,  because $a + \mathbf{1} \notin H_a$  according to \cref{lemma:notinH}. 
Thus, $H_a$ is an affine subspace of dimension $r-1$. Hence, there exist $h  \in \mathbb{R}^r \setminus \{\mathbf{0}\}$ and $b \in \mathbb{R}$ such that $H_a = \{x \in \mathbb{R}^r:h^Tx = b\}$. Assume that there exist two indices $i \neq j \in [r]$ such that $h_i > 0$ and $h_j < 0$. This would imply the existence of a non-zero vector $q \geq \mathbf{0}$ with $h^Tq = 0$. As $a \in H_a$ and $h^Ta = b$ we would get $h^T(a+q) = b$. However, as $q$ is non-zero and $q \geq \mathbf{0}$, $a + q \notin H_a$ according to \cref{lemma:notinH}, a contradiction. Thus, either $h  \geq 0$ or $h \leq 0$. As $ \{x \in \mathbb{R}^r:h^Tx = b\} =  \{x \in \mathbb{R}^r:(-h)^Tx = -b\}$, this proves the lemma.  
\end{proof}

The next lemma lies at the heart of our argument. 

\begin{lemma}
\label{lemma:quadratic}
For any maximal representable tuple $a$ and any $a' \in H_a$, there exist values $\alpha'_{ij} \in \mathbb{R}$ for $i \neq j \in [r]$ with $a' = a + \sum_{i \neq j \in [r]} \alpha'_{ij}w_{ij}$ such that for each $k \in [r]$, either $\alpha'_{ij}(w_{ij})_k \leq 0$ for each $i \neq j \in [r]$ or $\alpha'_{ij}(w_{ij})_k \geq 0$ for each $i \neq j \in [r]$. 
\end{lemma}
\begin{proof}
Let $a' \in H_a$ be arbitrary. Let $b \in H_a$ a vector such that for each $i \in [r]$, either $a_i \leq b_i \leq a_i'$ or $a_i' \leq b_i \leq a_i$ and there exist values $\beta_{ij} \in \mathbb{R}$ for $i \neq j \in [r]$ such that $b = a + \sum_{i \neq j \in [r]} \beta_{ij}w_{ij}$. 
Furthermore, for each $k \in [r]$, either $\beta_{ij}(w_{ij})_k \leq 0$ for each $i \neq j \in [r]$ or $\beta_{ij}(w_{ij})_k \geq 0$ for each $i \neq j \in [r]$. 
Note that such a vector $b$ always exists, as setting $b = a $ and all the $\beta_{ij}$'s to $0$ would fulfill all the criteria. 
We choose $b$ in such a way that the number of coordinates that $b$ and $a'$ disagree with is minimal. Note that showing $b = a'$ is equivalent to the statement of the lemma. For the sake of contradiction, assume that this is not the case. 

As $a',b \in H_a$, we also have $a + (a'-b) \in H_a$ and $a + (b-a') \in H_a$. 
If $a'_i \geq b_i$ for all $i\in [r]$, then $a' - b$ is a non-zero vector with $a'-b \geq \mathbf{0}$. 
Hence, according to \cref{lemma:notinH},  $a + (a' - b) \notin H_a$, a contradiction. 
Similarly, $a'_i \leq b_i$ for all $i \in [r]$ would also lead to a contradiction. 
Thus, we can conclude that there exist two indices $k,\ell \in [r]$ such that $b_{k}< a'_k$ and $b_{\ell} > a'_{\ell}$. Now, consider the vector $c = a + \sum_{i \neq j \in [r]} \gamma_{ij}w_{ij}$ where for $i \neq j \in [r]$ we define 
\[
\gamma_{ij} 
= \begin{cases} 
\beta_{k\ell} + \min \left(\frac{(a'_k - b_k)}{(w_{k\ell})_k}, \frac{(b_\ell - a'_\ell)}{(w_{\ell k})_\ell} \right) \text{ if $i = k$ and $j = \ell$} \\
\beta_{ij} \text{ otherwise}
\end{cases}
\]
We show that $c$ contradicts the assumption that $b$ is a vector that agrees with $a'$ on the maximum number of coordinates, among the vectors satisfying the properties stated in the beginning. 
We start by showing that for each coordinate $m \in [r]$, we either have $\gamma_{ij}(w_{ij})_m \leq 0$ for each $i \neq j \in [r]$ or $\gamma_{ij}(w_{ij})_m \geq 0$ for each $i \neq j \in [r]$. 
As we assume that this property holds for the vector $b$, we only need to show it for the coordinates $k$ and $\ell$. 

Recall that we have either $a_i \le b_i \le a_i'$ or $a_i' \le b_i \le a_i$ and that $b_k < a_k'$. 
This implies that $a_k \leq b_k < a'_k$. 
In particular, this implies that $\beta_{ij}(w_{ij})_k \geq 0$ for each $i \neq j \in [r]$. 
Thus, it remains to show that $\gamma_{k\ell}(w_{k \ell})_k \geq 0$, which is the case as $(w_{k \ell})_k > 0$ implies $\beta_{k\ell} \geq 0$ and therefore also $\gamma_{k \ell} \geq 0$, as  $\min \left(\frac{(a'_k - b_k)}{(w_{k\ell})_k}, \frac{(b_\ell - a'_\ell)}{(w_{\ell k})_\ell} \right) \geq 0$.  
Proceeding in the same manner, we get that $b_{\ell} > a'_{\ell}$ implies that $a'_{\ell} < b_{\ell} \leq a_{\ell}$. 
In particular, this implies that $\beta_{ij}(w_{ij})_\ell \leq 0$ for each $i \neq j \in [r]$.
Thus, it remains to show that $\gamma_{kl}(w_{k\ell})_{\ell} \leq 0$, which is the case as $(w_{k\ell})_\ell \leq 0$ and $\gamma_{k \ell} \geq 0$. 

Next, we show that for each $i \in [r]$, we either have $a_i \leq c_i \leq a'_i$ or $a_i \geq c_i \geq a'_i$. As $b_i = c_i$ for each $i \in [r] \setminus \{k,\ell\}$, we only need to show it for the coordinates $k$ and $\ell$. We have:

\begin{align*}
a_k &\leq b_k \leq  b_k +\min \left(\frac{(a'_k - b_k)}{(w_{k\ell})_k}, \frac{(b_\ell - a'_\ell)}{(w_{\ell k})_\ell} \right)(w_{k\ell})_k\\
&\leq b_k + \frac{(a'_k - b_k)}{(w_{k\ell})_k}(w_{k\ell})_k  = a'_k
\end{align*}
and 
\begin{align*}
a_\ell &\geq b_\ell \geq b_\ell +\min \left(\frac{(a'_k - b_k)}{(w_{k\ell})_k}, \frac{(b_\ell - a'_\ell)}{(w_{\ell k})_\ell} \right)(w_{k\ell})_\ell\\
&\geq b_{\ell} + \frac{(b_\ell - a'_\ell)}{(w_{\ell k})_\ell}(w_{k\ell})_\ell = a'_\ell
\end{align*}
and therefore $a_k \leq c_k \leq a'_k$ and $a_\ell \geq c_\ell \geq a'_\ell$, as desired. In the second line, we used the fact that $(w_{k\ell})_\ell = -(w_{\ell k})_\ell \leq 0$. Furthermore, note that if $\min \left(\frac{(a'_k - b_k)}{(w_{k\ell})_k}, \frac{(b_\ell - a'_\ell)}{(w_{\ell k})_\ell} \right) = \frac{(a'_k - b_k)}{(w_{k\ell})_k}$, then $c_k = a'_k$, and otherwise $c_\ell = a'_\ell$. Therefore, $c$ and $a'$ differ in a smaller number of coordinates than $b$ and $a'$, which is a contradiction. 
\end{proof}

We will use the following corollary of the above statement. 

\begin{corollary}
\label{cor:unitvectorsmallalphas}
Let $c = \max_{i \in [r]} 1/a_i$. For each unit vector $v \in \mathbb{R}^r$ with $a + v \in H_a$, there exist $\alpha_{ij}$'s with $v = \sum_{i \neq j \in [r]} \alpha_{ij}w_{ij}$ such that for each $k \in [r]$, either $\alpha_{ij}(w_{ij})_k \leq 0$ for each $i \neq j \in [r]$ or $\alpha_{ij}(w_{ij})_k \geq 0$ for each $i \neq j \in [r]$ and furthermore, $|\alpha_{ij}| \leq c$, for each $i \neq j \in [r]$. 
\end{corollary}

\begin{proof}
Let $v \in \mathbb{R}^r$ be a unit vector with $a + v \in H_a$. According to \cref{lemma:quadratic}, there exist $\alpha_{ij}$'s such that $v = \sum_{i \neq j \in [r]} \alpha_{ij}w_{ij}$ and for each $k \in [r]$, either $\alpha_{ij}(w_{ij})_k \leq 0$ for each $i \neq j \in [r]$ or $\alpha_{ij}(w_{ij})_k \geq 0$ for each $i \neq j \in [r]$. We show that $|\alpha_{ij}| \leq c$ for each $i \neq j \in [r]$. For the sake of contradiction, assume there exist $k \neq \ell \in [r]$ such that $|\alpha_{k \ell}| > c$. This implies
$$|\alpha_{k \ell} (w_{k \ell})_k | = \big| \alpha_{k \ell} \frac{a_k}{a_{k \ell}}\big| > \frac{1}{a_k} \cdot \frac{a_k}{a_{k\ell}}  > 1$$
and therefore:

\[\big|v_k\big| = \big|\sum_{i \neq j \in [r]} \alpha_{ij}(w_{ij})_k \big| = \sum_{i \neq j \in [r]} \big| \alpha_{ij}(w_{ij})_k \big| \geq \big|\alpha_{k \ell}(w_{k \ell})_k\big| > 1\]

The second inequality follows as for each $i \neq j \in [r]$, $\alpha_{ij}(w_{ij})_k$ has the same sign. This is a contradiction as $v$ is a unit vector and therefore $|v_k| \leq 1$.
\end{proof}

The main theorem now follows by carefully checking that the quadratic terms appearing when we generate $a'$ are always positive.

\begin{theorem}
\label{theorem:hyperplane_rep}
For any maximal representable tuple $a$, there 
exists an $\varepsilon > 0$ such that for any $a' \in H_a \cap B(a,\varepsilon)$, $a'$ is representable.
\end{theorem}\begin{proof}
Note that it is sufficient to prove the existence of an $\varepsilon > 0$, such that for any unit vector $v \in \mathbb{R}^r$ with $a + v \in H_a$ and every $0 \leq \delta < \varepsilon$, the tuple $a^\delta := a + \delta v$ is representable. 
As $v$ is a unit vector with $a + v \in H_a$, according to \cref{cor:unitvectorsmallalphas},  there exist $\alpha_{ij}$'s such that $v = a + \sum_{i \neq j \in [r]} \alpha_{ij}w_{ij}$ and, moreover, for each $k \in [r]$, we either have $\alpha_{ij}(w_{ij})_k \leq 0$ for each $i \neq j \in [r]$ or
$\alpha_{ij}(w_{ij})_k \geq 0$ for each $i \neq j \in [r]$, and $|\alpha_{ij}| \leq c:= max_{i \in [r]}1/a_i$. We have $a^\delta := a + \delta v = a + \sum_{i \neq j \in [r]} (\delta \alpha_{ij})w_{ij}$. 
Consider now $(a^\delta_{ij})$ with $a^\delta_{ij} = a_{ij} + \delta (\alpha_{ij} - \alpha_{ji})$. Note that $a^\delta_{ij} + a^\delta_{ji} = a_{ij} + a_{ji} \leq 1$ for each $\delta$. Furthermore, for each $\delta \geq 0$ and $i \neq j \in [r]$, $|a^\delta_{ij}-a_{ij}| \leq \delta(|\alpha_{ij}| + |\alpha_{ji}|) \leq 2 \delta c$. As $c$ only depends on $a$, there exists some $\varepsilon' > 0$, independent of $v$, such that for each $0 \leq \delta \leq \varepsilon'$, $(a^\delta_{ij})$ is a valid generator. 

Next, we show that there exists some $0 < \varepsilon < \eps'$, again independent of $v$, such that for each $0 \leq \delta \leq \varepsilon$, $(a^\delta_{ij})$ generates a tuple with each coordinate being at least as large as the corresponding coordinate in $a^\delta$. This implies that $a^\delta$ is representable, hence proving the claim.
To that end, note that for an arbitrary $k \in [r]$ we have

\begin{align*}
a^\delta_k &= \prod_{j \neq k} a^\delta_{kj} 
     = \prod_{j \neq k} (a_{kj} + \delta \cdot (\alpha_{kj} - \alpha_{jk})) \\
 &= a_k + \sum_{\ell \neq k} \delta \cdot (\alpha_{k\ell} - \alpha_{\ell k}) \prod_{j \notin \{k,\ell\}} a_{kj}\\
 &\quad + \delta^2 \sum_{\ell' \neq k} \sum_{\ell'' \not\in \{ k, \ell'\} } (\alpha_{k\ell'} - \alpha_{\ell' k}) \cdot  (\alpha_{k\ell''} - \alpha_{\ell'' k}) \prod_{j \notin \{k,\ell', \ell''\}} a_{kj} \\
 &\quad + \delta^3 \sum_{\ell' \neq k} \sum_{\ell'' \not\in \{ k, \ell'\} } \sum_{\ell'''  \not\in \{ k, \ell', \ell''\} } \ (\alpha_{k\ell'} - \alpha_{\ell' k}) \cdot  (\alpha_{k\ell''} - \alpha_{\ell'' k})\\ &\qquad\quad\cdot  (\alpha_{k\ell'''} - \alpha_{\ell''' k}) \prod_{j \notin \{k,\ell', \ell'',\ell'''\}} a_{kj} + \ldots
\end{align*}
First, we take a look at the term linear in $\delta$. We get:
\begin{align*}
&\sum_{\ell \neq k} \delta \cdot (\alpha_{k\ell} - \alpha_{\ell k}) \prod_{j \notin \{k,\ell\}} a_{kj} = \sum_{\ell \neq k} \delta \cdot \left( \alpha_{k\ell}\frac{a_k}{a_{k \ell}} + \alpha_{\ell k} \frac{-a_k}{a_{k\ell}} \right)\\
&= \sum_{\ell \neq k} \delta \alpha_{k\ell} (w_{k\ell})_k + \delta\alpha_{\ell k} (w_{\ell k})_k = \sum_{i \neq j} \delta \alpha_{ij}(w_{ij})_k = \delta v_k
\end{align*}
Next, we find a lower bound for the quadratic term.
Note that for each $\ell' \neq \ell'' \in [r] \setminus \{k\}$, we have:
\begin{align*}
&(\alpha_{k\ell'} - \alpha_{\ell'k}) \cdot (\alpha_{k\ell''} - \alpha_{\ell''k})\\
&= \left(\alpha_{k\ell'}(w_{k\ell'})_k \frac{a_{k\ell'}}{a_k} + \alpha_{\ell'k}(w_{\ell'k})_k \frac{a_{k\ell'}}{a_k} \right)\\
&\quad\cdot\left(\alpha_{k\ell''}(w_{k\ell''})_k \frac{a_{k\ell''}}{a_k} + \alpha_{\ell''k}(w_{\ell''k})_k \frac{a_{k\ell''}}{a_k} \right) \geq 0
\end{align*}
as for all $i \neq j \in [r]$ $\alpha_{ij}(w_{ij})_k \geq 0$, or for all $i \neq j \in [r]$ $\alpha_{ij}(w_{ij})_k \leq 0$. Thus, each summand in the quadratic term is non-negative. Let us define 
$$ u := \max_{\ell' \neq \ell'' \in [r] \setminus \{k\}} (\alpha_{k\ell'} - \alpha_{\ell'k}) \cdot (\alpha_{k\ell''} - \alpha_{\ell''k}) \geq 0$$
We can lower bound the quadratic term by:
\[
\delta^2 \sum_{\ell' \neq k} \sum_{\ell'' \not\in \{ \ell', k\}} (\alpha_{k\ell'} - \alpha_{\ell' k}) \cdot  (\alpha_{k\ell''} - \alpha_{\ell'' k}) \prod_{j \notin \{k,\ell', \ell''\}} a_{kj} \geq \delta^2 \cdot u \cdot a_k
\]
Next, we find a lower bound for each higher order term. Each such higher order term is the sum of expressions with the following form:
\begin{align*}
&\delta^t \prod_{s=1}^t (\alpha_{k\ell_s} - \alpha_{\ell_s k}) \prod_{j \notin \{k, \ell_1, \ldots, \ell_t\}} a_{kj}\\
&\geq - \big\lvert \delta^t \prod_{s=1}^t (\alpha_{k\ell_s} - \alpha_{\ell_s k}) \prod_{j \notin \{k, \ell_1, \ldots, \ell_t\}} a_{kj} \big\rvert\\
&\geq -\delta^t \cdot u \cdot (2c)^{t-2} \geq -u \cdot \delta^3 \cdot (2c)^r
\end{align*}
for some distinct $\ell_1, \ldots, \ell_t \in [r] \setminus \{k\}$ and some $t \in \mathbb{N}$ with $t \geq 3$. As there are at most $2^r$ such terms, we can conclude that:
\begin{align*}
a^\delta_k = \prod_{j \neq k} a^\delta_{kj} &\geq a_k + \delta v_k + \delta^2 \cdot u \cdot a_k - 2^r (u \cdot \delta^3 \cdot (2c)^r)\\
&\quad= a'(\delta)_k + u\delta^2(a_k - \delta \cdot 2^r \cdot (2c)^r) \geq a'(\delta)_k
\end{align*}
for $\delta \leq \frac{a_k}{2^r \cdot (2c)^r} := \varepsilon''$
with $\varepsilon''$ only depending on the tuple $(a_1, \ldots, a_n)$ and not the vector $v$. Setting $\varepsilon = \min(\varepsilon', \varepsilon'')$, we can conclude that for each $\delta$ with $0 \leq \delta \leq \varepsilon$, $a'(\delta)$ is a representable tuple. This concludes the proof.
\end{proof}

\begin{lemma}
	For each maximal representable tuple $a$, there exists a weakly locally supportive hyperplane for $\isnon$ containing $a$.   
\end{lemma}
\begin{proof}
	According to \cref{cor:hyperplane}, there exist $h \in \mathbb{R}^r$ with $h \geq \mathbf{0}$ and $b \in \mathbb{R}$ such that $H_a = \{x \in \mathbf{R}^r \colon h^Tx = b \}$. As $a \in H_a$, we have $h^Ta = b$. Let $\varepsilon' > 0$ such that for each $a' \in H_a \cap B(a,\varepsilon')$, $a'$ is a representable tuple. According to \cref{theorem:hyperplane_rep}, such an $\varepsilon'$ exists. We set $\varepsilon = \min(\varepsilon', \min_{i \in [r]} a_i) > 0$. Let $a'' \in B(a,\varepsilon)$ with $h^Ta'' \leq b$. As  $h^Ta'' \leq b$, there exists some $\delta \geq 0$ such that $h^Ta' = b$ with $a' := a'' + \delta h$. As $h^T(a - a') = b - b = 0$ and $a' - a'' = \delta h$, we can write $a-a'' = (a - a') + (a' - a'')$ with $(a-a')^T(a'-a'') = 0$. Thus, we can conclude that 
	$||a-a'|| \leq ||a-a''|| \leq \varepsilon$. 
	Hence, $a' \in H_a \cap B(a,\varepsilon)$ and therefore $a'$ is a representable tuple. As $h \geq 0$, $a'' = a' - \delta h \geq \mathbf{0}$ is also a representable tuple and therefore $a'' \in \srep$. Thus, $a'' \notin \isnon $, as desired.
\end{proof}

\section{Acknowledgments}
We thank Mohsen Ghaffari and anonymous referees for many helpful suggestions, and Danil Ko\v{z}evnikov and Yannic Maus for enlightening discussions about convexity.
This project has received funding from the European Research Council (ERC) under the European Union’s
Horizon 2020 research and innovation programme (grant agreement No. 853109).

\bibliography{references}

\newcommand{\etalchar}[1]{$^{#1}$}
\begin{thebibliography}{BFH{\etalchar{+}}16}

\bibitem[Alo91]{Alon1991}
Noga Alon.
\newblock {A Parallel Algorithmic Version of the Local Lemma}.
\newblock {\em Random Structures \& Algorithms}, 2(4):367--378, 1991.

\bibitem[AS08]{alonspencer}
Noga Alon and Joel~H. Spencer.
\newblock {\em The Probabilistic Method, Third Edition}.
\newblock Wiley-Interscience series in discrete mathematics and optimization.
  Wiley, 2008.

\bibitem[Bec91]{Beck1991}
J\'{o}zsef Beck.
\newblock {An Algorithmic Approach to the Lov\'{a}sz Local Lemma.}
\newblock {\em Random Structures \& Algorithms}, 2(4):343--365, 1991.

\bibitem[BFH{\etalchar{+}}16]{LLL_lowerbound}
Sebastian Brandt, Orr Fischer, Juho Hirvonen, Barbara Keller, Tuomo
  Lempi\"ainen, Joel Rybicki, Jukka Suomela, and Jara Uitto.
\newblock {A Lower Bound for the Distributed Lov\'asz Local Lemma}.
\newblock In {\em the Proceedings of the ACM-SIAM Symposium on Discrete
  Algorithms (SODA)}, 2016.

\bibitem[BMU19]{brandt2019sharp}
Sebastian Brandt, Yannic Maus, and Jara Uitto.
\newblock A sharp threshold phenomenon for the distributed complexity of the
  lov{\'{a}}sz local lemma.
\newblock In {\em Proceedings of the 2019 {ACM} Symposium on Principles of
  Distributed Computing, {PODC} 2019, Toronto, ON, Canada, July 29 - August 2,
  2019}, pages 389--398, 2019.

\bibitem[CKP16]{Chang2016a}
Yi{-}Jun Chang, Tsvi Kopelowitz, and Seth Pettie.
\newblock {An Exponential Separation between Randomized and Deterministic
  Complexity in the {LOCAL} Model}.
\newblock In {\em the Proceedings of the Symposium on Foundations of Computer
  Science (FOCS)}, pages 615--624, 2016.

\bibitem[CP17]{ChangHierarchy17}
Yi-Jun Chang and Seth Pettie.
\newblock {A Time Hierarchy Theorem for the LOCAL Model}.
\newblock In {\em the Proceedings of the Symposium on Foundations of Computer
  Science (FOCS)}, pages 156--167, 2017.

\bibitem[CPS17]{SuLLL2017}
Kai-Min Chung, Seth Pettie, and Hsin-Hao Su.
\newblock {Distributed Algorithms for the Lov{\'a}sz Local Lemma and Graph
  Coloring}.
\newblock {\em Distributed Computing}, 30(4):261--280, 2017.

\bibitem[CS00a]{Czumaj2000}
Artur Czumaj and Christian Scheideler.
\newblock {Coloring Non-uniform Hypergraphs: A New Algorithmic Approach to the
  General Lov\'{a}sz Local Lemma}.
\newblock In {\em the Proceedings of the ACM-SIAM Symposium on Discrete
  Algorithms (SODA)}, pages 30--39, 2000.

\bibitem[CS00b]{czumaj00algorithmic}
Artur Czumaj and Christian Scheideler.
\newblock A new algorithmic approach to the general {L}ov{\'a}sz local lemma
  with applications to scheduling and satisfiability problems.
\newblock In {\em Proceedings of the 32nd Annual ACM Symposium on Theory of
  Computing (STOC)}, pages 38--47, 2000.

\bibitem[EL75]{erdos75local}
Paul Erd{\H o}s and L{\'a}szl{\'o} Lov{\'a}sz.
\newblock Problems and results on 3-chromatic hypergraphs and some related
  questions.
\newblock {\em Infinite and finite sets}, 2(2):609--627, 1975.

\bibitem[EPS15]{elkin15matching}
Michael Elkin, Seth Pettie, and Hsin-Hao Su.
\newblock {$(2\Delta-1)$}-edge-coloring is much easier than maximal matching in
  the distributed setting.
\newblock In {\em Proceedings of the 26th Annual ACM-SIAM Symposium on Discrete
  Algorithms (SODA)}, pages 355--370, 2015.

\bibitem[FG17]{ManuelaLLL17}
Manuela Fischer and Mohsen Ghaffari.
\newblock {Sublogarithmic Distributed Algorithms for Lov{\'a}sz Local Lemma,
  and the Complexity Hierarchy}.
\newblock In {\em the Proceedings of the 31st International Symposium on
  Distributed Computing (DISC)}, pages 18:1--18:16, 2017.

\bibitem[FHK16]{fraigniaud16}
Pierre Fraigniaud, Marc Heinrich, and Adrian Kosowski.
\newblock {Local Conflict Coloring}.
\newblock In {\em the Proceedings of the Symposium on Foundations of Computer
  Science (FOCS)}, pages 625--634, 2016.

\bibitem[Gha16]{GhaffariImproved16}
Mohsen Ghaffari.
\newblock {An Improved Distributed Algorithm for Maximal Independent Set}.
\newblock In {\em the Proceedings of the ACM-SIAM Symposium on Discrete
  Algorithms (SODA)}, pages 270--277, 2016.

\bibitem[GHK18]{newHypergraphMatching}
Mohsen Ghaffari, David~G. Harris, and Fabian Kuhn.
\newblock {On Derandomizing Local Distributed Algorithms}.
\newblock In {\em the Proceedings of the Symposium on Foundations of Computer
  Science (FOCS)}, pages 662--673, 2018.

\bibitem[HSS10]{haeupler10constructive}
Bernhard Haeupler, Barna Saha, and Aravind Srinivasan.
\newblock New constructive aspects of the {L}ov{\'a}sz local lemma.
\newblock In {\em Proceedings of the 51st Annual Symposium on Foundations of
  Computer Science (FOCS)}, pages 397--406, 2010.

\bibitem[Lin92]{linial92}
N.~Linial.
\newblock Locality in distributed graph algorithms.
\newblock {\em SIAM Journal on Computing}, 21(1):193--201, 1992.

\bibitem[LMR99]{leighton99fast}
Tom Leighton, Bruce Maggs, and Andr{\'e}a~W. Richa.
\newblock Fast algorithms for finding {O}(congestion + dilation) packet routing
  schedules.
\newblock {\em Combinatorica}, 19(3):375--401, 1999.

\bibitem[Mos08]{Moser08}
Robin~A. Moser.
\newblock Derandomizing the {L}ov{\'a}sz local lemma more effectively.
\newblock {\em CoRR}, abs/0807.2120, 2008.

\bibitem[Mos09]{Moser09}
Robin~A. Moser.
\newblock A constructive proof of the {L}ov{\'a}sz local lemma.
\newblock In {\em Proceedings of the 41st Annual ACM Symposium on Theory of
  Computing (STOC)}, pages 343--350, 2009.

\bibitem[MR98]{Molloy1998}
Michael Molloy and Bruce Reed.
\newblock {Further Algorithmic Aspects of the Local Lemma}.
\newblock In {\em the Proceedings of the ACM-SIAM Symposium on Discrete
  Algorithms (SODA)}, pages 524--529, 1998.

\bibitem[MT10]{MoserTardos10}
Robin~A. Moser and G\'{a}bor Tardos.
\newblock {A Constructive Proof of the General Lov\'{a}sz Local Lemma}.
\newblock {\em J. ACM}, pages 11:1--11:15, 2010.

\bibitem[Pel00]{peleg00}
D.~Peleg.
\newblock {\em Distributed Computing: A Locality-Sensitive Approach}.
\newblock SIAM, 2000.

\bibitem[RG20]{Rozhon19}
Vaclav Rozhon and Mohsen Ghaffari.
\newblock Polylogarithmic-time deterministic network decomposition and
  distributed derandomization.
\newblock In {\em the Proceedings of the ACM-SIAM Symposium on Discrete
  Algorithms (SODA)}, 2020.

\bibitem[She85]{Shearer85}
James~B. Shearer.
\newblock On a problem of {S}pencer.
\newblock {\em Combinatorica}, 5(3):241--245, 1985.

\bibitem[Spe77]{Spencer77}
Joel Spencer.
\newblock Asymptotic lower bounds for ramsey functions.
\newblock {\em Discrete Mathematics}, 20:69--76, 1977.

\bibitem[Sri08]{Srinivasan2008}
Aravind Srinivasan.
\newblock {Improved Algorithmic Versions of the Lov\'{a}sz Local Lemma}.
\newblock In {\em the Proceedings of the ACM-SIAM Symposium on Discrete
  Algorithms (SODA)}, pages 611--620, 2008.

\bibitem[Val75]{valentine}
Frederick A. (Frederick~Albert) Valentine.
\newblock {\em Convex sets}.
\newblock Huntington, N.Y. : R. E. Krieger Pub. Co, 1975.
\newblock Reprint of the ed. published by McGraw-Hill in series: McGraw-Hill
  series in higher mathematics.

\end{thebibliography}
\bibliographystyle{alpha}

\end{document}